\documentclass[11pt]{article}

\usepackage{lineno}
\usepackage{amsmath}
\usepackage{amssymb}
\usepackage{amsthm}
\usepackage{graphicx}
\usepackage{enumerate}
\usepackage{cite}
\usepackage{fullpage}
\usepackage[usenames]{color}
\usepackage[nodayofweek]{datetime}

\usepackage{algpseudocode}

\algrenewcommand\algorithmicprocedure{\textbf{Algorithm}}

\newcommand{\Real}{\ensuremath{\mathbb{R}}}
\newcommand{\Plane}{\ensuremath{\mathbb{R}^2}}

\newcommand{\Poly}{\ensuremath{\mathcal{P}}}
\newcommand{\SPM}{\ensuremath{\mathsf{SPM}}}
\newcommand{\SPT}{\ensuremath{\mathsf{SPT}}}
\newcommand{\Vis}{\ensuremath{\mathsf{VR}}}

\newcommand{\plf}{\ensuremath{\mathrm{len}}}

\newcommand{\bd}{\ensuremath{\partial}}

\newcommand{\seg}{\overline}
\newcommand{\cen}{\ensuremath{\mathrm{cen}}}

\newcommand{\prad}{\Phi}
\newcommand{\rad}{\mathrm{rad}}

\newcommand{\dist}{\mathrm{d}}

\newcommand{\DecompSPM}{\mathcal{A}_\mathsf{SPM}}
\newcommand{\DecompSPT}{\mathcal{A}_\mathsf{SPT}}

\newtheorem{lemma}{Lemma}
\newtheorem{theorem}{Theorem}

\makeatletter
\newbox\ProofSym
\setbox\ProofSym=\hbox{%
\unitlength=0.18ex%
\begin{picture}(10,10)
\put(0,0){\framebox(9,9){}} \put(0,3){\framebox(6,6){}}
\end{picture}}
\renewenvironment{proof}[1][Proof.]{\O@proof{#1}}{\O@endproof}
\def\O@proof#1{\trivlist
   \@topsep\z@\@topsepadd\smallskipamount%
   \@ifstar{\item[]}{\item[\hskip\labelsep\it #1 ]}}
\def\O@endproof{\hfill\copy\ProofSym\linebreak[3mm]\endtrivlist}
\makeatother

\title{Computing the Geodesic Centers of a Polygonal Domain%
\thanks{%
A preliminary version of this paper was presented at
the \emph{26th Canadian Conference on Computational Geometry (CCCG'14)}~\cite{bko-cgcpd-14}.
Work by S.W. Bae was supported by Basic Science Research Program through the National Research Foundation
of Korea (NRF) funded by the Ministry of Science, ICT \& Future Planning (2013R1A1A1A05006927).
Work by M. Korman was partially supported by the ELC project (MEXT KAKENHI No. 24106008).
Work by Y. Okamoto was partially supported by Grant-in-Aid for Scientific Research from Ministry of Education,
Science and Culture, Japan and Japan Society for the Promotion of Science,
and the ELC project (Grant-in-Aid for Scientific Research on Innovative Areas, MEXT Japan).
}
}

\author{%
Sang Won Bae\thanks{%
Department of Computer Science, Kyonggi University, Suwon, Korea.
Email: \texttt{swbae@kgu.ac.kr}
}
\and %
Matias Korman\thanks{%
Graduate School of Information and Science, Tohoku University, Sendai, Japan. Email: \texttt{mati@dais.is.tohoku.ac.jp}
}
\and %
Yoshio Okamoto\thanks{%
Department of Communication Engineering and Informatics,
University of Electro-Communications, Chofu, Tokyo, Japan.
Email: \texttt{okamotoy@uec.ac.jp}} }

\begin{document}

\maketitle

\begin{abstract}
We present an algorithm that computes the geodesic center of a given polygonal domain.
The running time of our algorithm is $O(n^{12+\epsilon})$ for any $\epsilon>0$,
where $n$ is the number of corners of the input polygonal domain.
Prior to our work, only the very special case where a simple polygon is given as input
has been intensively studied in the 1980s,
and an $O(n \log n)$-time algorithm is known by Pollack et al.
Our algorithm is the first one that can handle general polygonal domains
having one or more polygonal holes.
\end{abstract}

\section{Introduction} \label{sec:intro}
The \emph{diameter} and \emph{radius} of  a compact shape
are among the most natural and fundamental parameters describing and summarizing
the shape itself.
In this paper, we study these quantities for a \emph{polygonal domain} $\Poly$,
that is, a polygon having $h \geq 0$ holes.
More specifically, a polygonal domain is a connected and compact subset of $\Plane$
whose boundary consists of $h+1$ simple closed polygonal curves.
In regard to a metric $d$ on $\Poly$, the diameter of $\Poly$ is defined
to be the maximum distance over all pairs of points in the $\Poly$,
that is, $\max_{p,q \in \Poly} d(p, q)$,
while the radius is defined to be the min-max value $\min_{p\in \Poly} \max_{q\in \Poly} d(p, q)$.
A pair of points in $\Poly$ realizing the diameter is called
a \emph{diametral pair}, and a \emph{center} is defined to be a point $c\in \Poly$
such that $\max_{q\in \Poly} d(c, q)$ is equal to the radius.
Among common metrics on a polygonal domain $\Poly$,
we consider the \emph{geodesic distance} $\dist(p, q)$ for $p,q\in \Poly$
that measures the Euclidean length of a shortest path
that connects $p$ and $q$ and stays inside $\Poly$.
The diameter and the radius of a polygonal domain $\Poly$ with respect to the geodesic distance $\dist$
are often called  \emph{geodesic diameter} and \emph{geodesic radius} of $\Poly$, respectively.

The problem of computing the geodesic diameter and radius of a simple polygon
(i.e., a polygonal domain with no holes)
has been intensively studied in computational geometry since the early 80s.
For the geodesic diameter problem, Chazelle~\cite{c-tpca-82} gave the first algorithm
whose running time was $O(n^2)$, where $n$ denotes the number of vertices or corners\footnote{%
A corner of a polygon usually indicates a vertex to which two incident edges
form an angle that is not $180^\circ$.}
of the input polygon.
This was afterwards improved to $O(n\log n)$ time by Suri~\cite{s-cgfnsp-89},
and finally to linear time by Hershberger and Suri~\cite{hs-msspm-97}.
For the geodesic radius of a simple polygon,
the first algorithm was given by Asano and Toussaint~\cite{at-cgcsp-85},
and its running time was $O(n^4\log n)$-time.
Later Pollack, Sharir, and Rote~\cite{psr-cgcsp-89} improved it to $O(n\log n)$ time.
Very recently, an optimal $O(n)$-time algorithm for the geodesic radius
of a simple polygon is presented by Ahn et al.~\cite{abbcko-ltagcsp-15}.

The case in which the domain has one or more holes is much less understood.
To the best of our knowledge, the only known result is a companion paper
in which an algorithm that computes the geodesic diameter of a polygonal domain
with $n$ corners and $h$ holes in $O(n^{7.73})$ or $O(n^7(\log n + h))$ time~\cite{bko-gdpd-13} is given.
As for computing the radius, no algorithm was known prior to this work,
even though the problem has been remarked repeatedly
as an important open problem~\cite[Open Problem 6]{m-gspno-00}.

The main difference between simple polygons and general domains lies on
the difficulty to determine and discretize the search space.
A key tool often used in these problems is, given a point $p\in\Poly$,
compute a farthest neighbor of $p$,
a point of $\Poly$ that is farthest away from $p$.
It is well known that in simple polygons, every farthest neighbor of any point $\Poly$
should be a corner of $\Poly$~\cite{at-cgcsp-85}.
This implies that the geodesic diameter of any simple polygon can only be determined
by two of its corners.
In particular, the problem is now reduced on how to efficiently search
among the $O(n^2)$ candidates that can potentially determine the geodesic diameter,
so one could try any bruteforce search on them.
The geodesic radius of a simple polygon can also be handled in a similar way:
Even though the corresponding center itself may be an interior point of $\Poly$,
its farthest neighbors are all corners.

For general polygonal domains, unfortunately, this is not the case any more.
A farthest neighbor of a point in a polygonal domain $\Poly$ having one or more holes
may not be a corner of $\Poly$, and even can be an interior point of $\Poly$.
This makes things complicated; the geodesic diameter can be determined
by two interior points, as shown in~\cite{bko-gdpd-13}.
This difference mainly causes the huge gap, $O(n)$ and $O(n^{7.73})$,
in the computational complexity of computing the geodesic diameter
between simple polygons~\cite{hs-msspm-97} and general domains~\cite{bko-gdpd-13}.

In this paper, we present an algorithm that, in $O(n^{12+\epsilon})$ time,
computes the geodesic radius and center.
At a glance, the time complexity might seem very high,
but it is comparable to the currently best algorithms for computing the geodesic diameter.
Indeed, a crucial observation for the diameter algorithm is
that there are at least five shortest paths between the two points determining
the geodesic diameter if the two points lie in the interior of $\Poly$~\cite{bko-gdpd-13}.
This observation leads to a bounded number of candidates for diametral pairs.
In Section~\ref{sec:neighbor}, we show that the geodesic radius and center
sometimes involves nine paths to determine.
This enlarges the search space considerably, thus a larger running time is somehow expected.

The rest of the paper is organized as follows.
After introducing preliminary definitions and concepts in Section~\ref{sec:pre},
we list geometric observations in Section~\ref{sec:neighbor}
that will be the base of our algorithm described in Section~\ref{sec:alg}.
Finally, Section~\ref{sec:conclusion} concludes the paper with possible lines of future research.

\section{Preliminaries} \label{sec:pre}
Throughout the paper, we frequently use several topological concepts such as
open and closed subsets, neighborhoods, and the boundary $\bd A$
of a set $A$;
unless stated otherwise, all of them are derived with respect to
the standard topology on $\Real^d$ with the Euclidean norm $\|\cdot\|$ for fixed $d\geq 1$.
We also denote the straight line segment joining two points $a, b \in \Plane$ by $\seg{ab}$.

A \emph{polygonal domain} $\Poly$ with $h$ holes and $n$ corners\footnote{%
We reserve the term ``vertex'' for a 0-dimensional face of subdivisions of a certain space.}
is a connected and closed subset of $\Plane$ with $h$ pairwise disjoint holes. Each hole is a simple polygon contained in $\Poly$. Thus, the boundary $\bd \Poly$ of $\Poly$ consists of $h+1$ simple closed polygonal chains, and overall $n$ line segments. Each of the holes (and the outer boundary of $\Poly$) is regarded as an \emph{obstacle} that feasible paths in $\Poly$ are not allowed to cross. The \emph{geodesic distance} $\dist(p,q)$ between any two points $p,q$
in a polygonal domain $\Poly$ is defined to be the Euclidean length of a shortest feasible path
between $p$ and $q$, where the \emph{length} of a path is the sum of the Euclidean lengths of its segments. It is well known~\cite{m-spaop-96} that there always exists a shortest feasible path between any two points $p, q \in \Poly$, and the geodesic distance function $\dist(\cdot, \cdot)$ is thus well defined.

The \emph{geodesic radius} $\rad(\Poly)$ of $\Poly$ is defined to be the min-max quantity:
 \[ \rad(\Poly) = \min_{p\in \Poly} \max_{q \in \Poly} \dist(p,q).\]
A \emph{geodesic center} of $\Poly$ is a point $c \in \Poly$ such that
 \[ \max_{q\in\Poly} \dist(c, q) = \rad(\Poly).\]

The set of all geodesic centers of $\Poly$ is denoted by $\cen(\Poly)$.
The purpose of this paper is to describe the first algorithm that exactly
computes the geodesic radius $\rad(\Poly)$ and centers $\cen(\Poly)$ of a given polygonal domain $\Poly$.

\subsection{Shortest path trees and shortest path maps}
\begin{figure*}[tb]
\centering
\includegraphics[width=.9\textwidth]{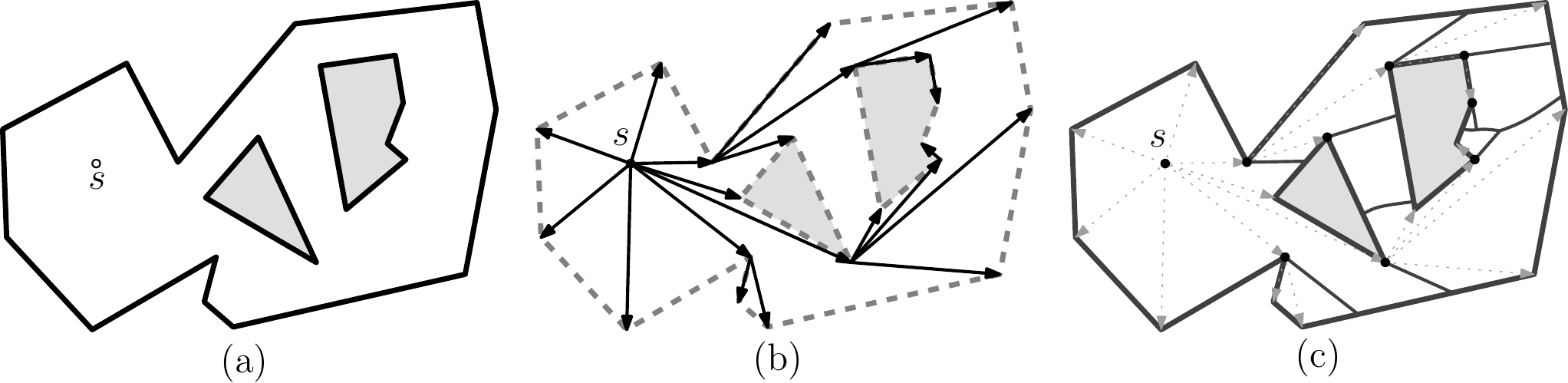}
\caption{(a) A polygonal domain $\Poly$ with two holes and a source point $s\in \Poly$.
         (b) The shortest path tree $\SPT(s)$ on $V \cup \{s\}$ with root $s$ (edges are directed
         towards descendants).
         (c) The shortest path map $\SPM(s)$ (depicted by solid segments). Corners $v \in V$ with non-empty region $\sigma_s(v)$ are marked by black dots.}
\label{fig:spm}
\end{figure*}

Let $V$ be the set of all corners of $\Poly$
and $\pi$ be a shortest path between any two points $s,t\in\Poly$.
This path $\pi$ is a polygonal chain that makes turns only at corners $V$ of $\Poly$~\cite{m-spaop-96}.
We represent $\pi$ by the sequence of traversed corners: $\pi = (s, v_1, \ldots, v_k, t)$
for some $v_1, \ldots, v_k\in V$.
Note that $k$ may be zero;
in this case, the shortest path $\pi$ is the segment $\seg{st}$ connecting the two endpoints,
and thus $\dist(s,t)=\|s-t\|$.
If two paths (with possibly different endpoints) induce the same sequence of corners
$(v_1, \ldots, v_k)$,
then they are said to have the same \emph{combinatorial structure}.

Given a source point $s\in \Poly$,
the \emph{shortest path tree} $\SPT(s)$ of $s$ is
a tree spanning $V \cup \{s\}$ embedded in $\Poly$
such that the unique path in $\SPT(s)$ from $s$ to each corner of $\Poly$
is a shortest path in $\Poly$. See for example \figurename~\ref{fig:spm}(b).

The \emph{shortest path map} $\SPM(s)$ of a fixed $s\in\Poly$ is a
decomposition of $\Poly$ into cells such that points in the same cell can be reached
from $s$ by shortest paths of the same combinatorial structure.
See \figurename~\ref{fig:spm}(c).
Each cell $\sigma_s(v)$ of $\SPM(s)$ is associated with a corner $v\in V$
which is the last corner of the shortest path $\pi$ from $s$ to any $t$ in the cell $\sigma_s(v)$.
Note that the path $\pi$ goes along the path in $\SPT(s)$ to $v$ and then reaches $t$ along $\seg{vt}$.
We also define the cell $\sigma_s(s)$ as the set of points $t\in\Poly$
such that $\seg{st} \subset \Poly$, i.e., $\dist(s, t) = \|s - t\|$.

Edges of $\SPM(s)$ either belong to $\bd\Poly$ or
are arcs on the boundary of two incident cells $\sigma_s(v_1)$ and $\sigma_s(v_2)$
determined by two corners $v_1, v_2 \in V\cup\{s\}$.
Edges of the second kind are hyperbolic arcs if $v_1$ and $v_2$ are not adjacent in $\SPT(s)$.
Moreover, there are two different shortest paths from $s$ to any point on an edge of $\SPM(s)$,
one via $v_1$ and the other via $v_2$.

Vertices of $\SPM(s)$ are either corners of $\Poly$, endpoints of an edge of the second kind above,
or a point $p \in \Poly$ incident to at least three faces $\sigma_s(v_1), \sigma_s(v_2), \sigma_s(v_3)$
for some corners $v_1,v_2,v_3 \in V \cup \{s\}$, yielding three different shortest paths from $s$.
Depending on which of the three cases it falls into,
each vertex of $\SPM(s)$ admits either $1$, $2$, or more different shortest paths from $s$ to the vertex, respectively.

The shortest path map $\SPM(s)$ 
has $O(n)$ cells, edges, and vertices in total,
and can be computed in $O(n\log n)$ time using $O(n\log n)$ working space~\cite{hs-oaespp-99}.
For more details on shortest path maps, we refer to~\cite{m-spaop-96, hs-oaespp-99, m-gspno-00}.


\subsection{Path-length functions}
For any point $p\in\Poly$, we define its \emph{visibility region}
as the set $\Vis(p)$ of all points $q\in\Poly$ such that $\seg{pq}\subset \Poly$,
that is, points $q$ that \emph{sees} $p$.

Let $\pi$ be a shortest path from $s$ to $t$ for $s, t\in \Poly$.
If $\pi \neq \seg{st}$, then there are two corners $u,v\in V$ such that $u$ and $v$ are the first and last corners
along $\pi$ from $s$ to $t$, respectively.
Here, the path $\pi$ is formed to be the union of
$\seg{su}$, $\seg{vt}$ and a shortest path from $u$ to $v$.
Note that $u$ and $v$ are not necessarily distinct.
In order to realize such a path, $s$ must see $u$ and $t$ must see $v$,
that is, $s\in \Vis(u)$ and $t\in \Vis(v)$,

We define the \emph{path-length function} $\plf_{u,v}\colon \Vis(u)\times \Vis(v) \to \Real$
for any fixed pair of corners $u,v\in V$ to be
 \[\plf_{u,v}(s,t) := \|s-u\| + \dist(u,v) + \|v-t\|.\]
That is, $\plf_{u,v}(s,t)$ represents the length of paths from $s$ to $t$
that have a common combinatorial structure; going straight from $s$ to $u$,
following a shortest path from $u$ to $v$, and going straight to $t$.
Also, unless $\dist(s,t) = \|s-t\|$ (equivalently, $s\in \Vis(t)$),
the geodesic distance $\dist(s,t)$ can be expressed as the pointwise minimum of path-length functions.

\[ \dist(s,t) = \min_{u \in \Vis(s)\cap V,~ v\in \Vis(t)\cap V} \plf_{u,v}(s,t).\]

By definition of shortest path map $\SPM(s)$ and its cells $\sigma_s(v)$,
if $t \in \sigma_s(v)$ for some $v\in V$, then we have
$\dist(s, t) = \plf_{u, v}(s, t)$,
where $u \in V$ denotes the first corner along the shortest path from $s$ to $v$,
or equivalently, along the path from $s$ to $v$ in $\SPT(s)$.

%

\section{Farthest Neighbors and Geodesic Centers} \label{sec:neighbor}
In this section we introduce several tools that will be useful for discretizing the search space. 
For any point $p \in \Poly$, we let $\prad(p)$ be the maximum geodesic distance we can obtain
when we fix one point as $p$, that is,
\[ \prad(p) := \max_{q \in \Poly} \dist(p, q).\]
We call a point $q\in \Poly$ a \emph{farthest neighbor} of $p \in \Poly$
if $\dist(p, q) = \prad(p)$.

Observe that the geodesic radius of $\Poly$ is the minimum possible value of $\prad(p)$
over all $p\in \Poly$, that is,
\[ \rad(\Poly) = \min_{p\in \Poly} \prad(p),\]
and a point that minimizes $\prad(p)$ is a geodesic center of $\Poly$.

The following lemma gives us a way of computing farthest neighbors.
Recall that each vertex of the shortest path map $\SPM(p)$ for $p\in\Poly$
is either a corner of $\Poly$,
an endpoint of an edge lying on the boundary $\bd \Poly$,
or a point in the interior of $\Poly$ that is incident to three edges.
\begin{lemma} \label{lemma:farthest_neighbor}
 For any point $p\in \Poly$,
 any farthest neighbor of $p$ in $\Poly$ is a vertex of $\SPM(p)$.
\end{lemma}
\begin{proof}
Suppose for the sake of contradiction that there exists a farthest neighbor  $q\in \Poly$ of $p$ that is not a vertex of $\SPM(p)$.
Then, there exists a sufficiently short line segment $L$ such that
$L$ is contained in the closure of some cell $\sigma_p(v)$ of $\SPM(p)$ for some $v \in V\cup \{p\}$
and contains $q$ in its relative interior.
This is always true even if $q$ lies on an edge of $\SPM(p)$
since every edge of the shortest path map is 
either straight or hyperbolic.

Then, the function $f(x)=\dist(p,x)$ for $x\in L$ is represented as
$f(x) = \plf_{u,v}(p, x) = \|p-u\| + \dist(u, v) + \|v - x\|$ for some $u\in V\cup\{p\}$.
Observe that the function $f$ is convex on $L$ and has no plateau along its graph.
Since $q$ lies in the relative interior of $L$,
there always exists a point $y\in L$ such that $f(y)>f(q)$,
which contradicts the assumption that $q$ is a farthest neighbor of $p$.
\end{proof}

This observation is analogous to the fact that
farthest neighbors of any point in a simple polygon are its corners~\cite{psr-cgcsp-89}.
However, vertices of shortest path maps $\SPM(p)$ may lie in the interior of $\Poly$.
This means that a geodesic radius and center may be determined by interior points,
whereas this never happens for simple polygons.


\begin{figure}[tb]
\centering
\includegraphics[width=.5\textwidth]{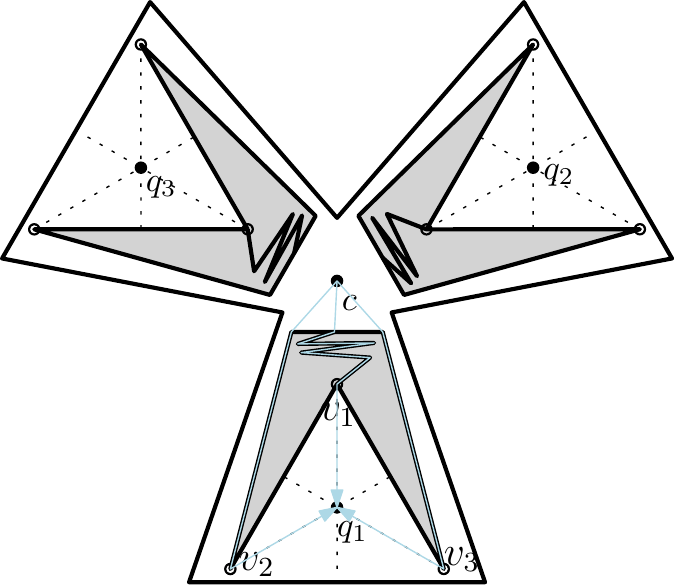}
\caption{A polygonal domain instance with a unique interior geodesic center $c$.
 The three farthest neighbors $q_1$, $q_2$, and $q_3$ of the center $c$
 lie in the interior of the domain as well.
 Observe that there are three distinct shortest paths between the center $c$
 and each of its farthest neighbors, and thus $9$ shortest paths of equal length in total.}
\label{fig:intcenter}
\end{figure}

\figurename~\ref{fig:intcenter} illustrates an example polygonal domain
such that its unique geodesic center $c$ and its farthest neighbors lie in its interior.
The domain shown in \figurename~\ref{fig:intcenter} consists of
three identical regions arranged in a symmetric way:
each part contains two holes that almost fit together forming a very narrow corridor between them.
We claim that $c$ is the unique geodesic center and $c$ has exactly three furthest neighbors:
$q_1$, $q_2$, and $q_3$.

Observe first that each $q_i$ is a vertex of the shortest path map $\SPM(c)$ for $c$.
Each vertex of $\SPM(c)$ lying in the interior admits three distinct shortest paths from $c$.
(Moreover, because of the way the regions are defined,
each $q_i$ is also the farthest neighbor of $c$ within the corresponding region.
Due to symmetry in the construction, we observe that no point other than $c$ can be closer
to all of $q_1$, $q_2$, and $q_3$ at the same time.
Thus, the point $c$ is the only geodesic center of this polygonal domain.


Note that this construction is slightly degenerate since it has a few symmetries.
However, such degeneracies can be removed by a small perturbation on the location of the corners.
This is possible because the center $c$ is ``stable'' in the sense that
a sufficiently small perturbation on the corners of the domain will only imply
a small change in the location of $c$ and points $q_i$.

\section{Algorithm} \label{sec:alg}
In this section, we describe our algorithm for computing the radius $\rad(\Poly)$ and
all centers $\cen(\Poly)$ of an input polygonal domain $\Poly$.
Recall that the problem of computing the radius and center can be seen
as a minimization problem under the objective function $\prad$ over $\Poly$.
Thus, our approach is to decompose $\Poly$ into cells, and find candidate centers in each cell.

For any subset $\sigma\subseteq \Poly$ of the domain $\Poly$,
we call the minimum value of $\prad(p)$ over $p\in \sigma$
the \emph{$\sigma$-constrained geodesic radius},
and each point in $\sigma$ that attains the minimum
is a \emph{$\sigma$-constrained geodesic center}.
Clearly, in any decomposition $\{\sigma_1, \sigma_2, \ldots\}$ of $\Poly$,
the geodesic radius is the minimum of $\sigma_i$-constrained geodesic radii over all $i$,
and the points that attain the minimum value form the geodesic centers $\cen(\Poly)$ of $\Poly$.

In this paper, we use the \emph{SPM-equivalence} decomposition.
This decomposition subdivides a polygonal domain $\Poly$ into cells
such that for all points $s$ in a common cell $\sigma$ of $\DecompSPM$,
their shortest path maps $\SPM(s)$ are topologically equivalent.
More precisely, two shortest path maps $\SPM(s_1)$ and $\SPM(s_2)$ are said to be \emph{topologically equivalent}
if their underlying labeled plane graphs are isomorphic.
This structure was introduced by  Chiang and Mitchell~\cite{cm-tpespqp-99}
as a means to devise efficient data structures that support two-point queries for Euclidean shortest paths.
In their work, they show that the decomposition $\DecompSPM$ has $O(n^{10})$ complexity
and can be computed in $O(n^{10} \log n)$ time.



An additional property of this subdivision (also shown by Chiang and Mitchell~\cite{cm-tpespqp-99}) is that,
for any cell $\sigma$ of $\DecompSPM$, the elements of $\SPM(s)$ (i.e., vertices and edges)
can be explicitly described by algebraic functions of $s \in \sigma$.
In this manner, the shortest path map $\SPM(s)$ for any point $s \in \sigma$
can be parameterized within a fixed cell $\sigma$.

Let $\sigma$ be any cell of $\DecompSPM$.
Since all shortest path maps $\SPM(s)$ within $s\in \sigma$ are topologically equivalent,
they must all have the same number $m$ of vertices.
We are particularly interested in coordinates of the vertices $v_1, v_2, \ldots, v_m$ of $\SPM(s)$
as functions of $s\in\sigma$.
Recall that the vertices $v_i$ of $\SPM(s)$ must include the corners $V$ of $\Poly$;
thus, if $v_i\in V$ for $i=1,2,\ldots,m$,
then $v_i(s)$ will be a constant function that maps to a unique corner of $\Poly$.

For $i=1,2,\ldots,m$, we define the function $f_i\colon \sigma \rightarrow \Real$ to be
$f_i(s)= \dist(s, v_i(s))$ for $s \in \sigma$.
That is, this function maps $s$ to the geodesic distance from $s$ to $v_i(s)$.
We then consider the upper envelope $\max_i f_i(s)$ of the $m$ functions,
which maps $s$ to its maximum geodesic distance over all the vertices $v_i(s)$ of $\SPM(s)$.
By Lemma~\ref{lemma:farthest_neighbor},
the farthest neighbors of $s$ must be among the $v_i(s)$, and it thus holds that
\[ \prad(s) = \max_{i=1,2,\ldots, m} f_i(s).\]
In order to find the $\sigma$-constrained geodesic radius and center,
it suffices to compute and search the upper envelope of the $m$ functions $f_i$.

In order to obtain an explicit expression of $f_i(s)$, we observe the following property.
\begin{lemma} \label{lemma:whole_vis}
 For any cell $\sigma$ of $\DecompSPM$ and $i\in\{1,2,\ldots, m\}$, one of the following holds
 for all $s\in \sigma$:
 $\sigma \subseteq \Vis(v_i(s))$ or $\sigma \cap \Vis(v_i(s))=\emptyset$.
\end{lemma}
\begin{proof}
This follows from the fact that shortest paths are topologically equivalent within $\sigma$.
If $s\in \sigma$ sees $v_i(s)$, then the shortest path $\pi$ from $s$ to $v_i(s)$ is $\seg{sv_i(s)}$.
Since the shortest path $\pi'$ from any other $s'\in \sigma$ to $v_i(s')$
must have the same combinatorial structure as $\pi = \seg{sv_i(s)}$ from $s$ to $v_i(s)$,
it holds that $\pi' = \seg{s'v_i(s')}$ for any $s' \in \sigma$.
\end{proof}

Lemma~\ref{lemma:whole_vis} shows that the visibility for the vertex $v_i(s)$ is preserved within cell $\sigma$.
Hence, we can simply say that $v_i(s)$ is always \emph{visible} from $\sigma$
or never visible from $\sigma$.
Since corners of $\Poly$ are also vertices of $\SPM(s)$,
they are also always or never visible from $\sigma$.


\begin{lemma} \label{lemma:visible_vertex_spm}
 For any cell $\sigma$ of $\DecompSPM$ and $i\in\{1,2,\ldots, m\}$,
 if vertex $v_i$ is visible from $\sigma$, then
 $v_i$ is a corner $v\in V$ of $\Poly$.
\end{lemma}
\begin{proof}
If $v_i$ is always visible from $\sigma$,
then the shortest path from $s$ to $v_i(s)$ is just the straight line segment $\seg{sv_i(s)}$
and therefore is unique.
However, as discussed in Section~\ref{sec:pre},
the only way in which a vertex of $\SPM(s)$ has a single shortest path from $s$ is
when it is a corner of $\Poly$.
\end{proof}

\begin{figure}[t]
\centering
\fbox{\begin{minipage}{.99\textwidth}
\begin{algorithmic}[1]
\Procedure{GeodesicCenter}{$\Poly$}
\State Compute the SPM-equivalence decomposition $\DecompSPM$ of $\Poly$.
 \For{each cell $\sigma$ of $\DecompSPM$}
 \State Specify the combinatorial structure of the shortest path maps $\SPM(s)$ for $s\in \sigma$.
 \State Identify the parameterized equations of the vertices of $\SPM(s)$.
 \State Let $v_1(s), \ldots, v_m(s)$ be the parameterized points identified by the above step.
 \State Let $f_i(s) := \dist(s, v_i(s))$ be the $m$ bivariate functions for $s\in \sigma$.
 \State Compute the upper envelope $\mathcal{U}_\sigma$ of the $m$ graphs $\{z= f_i(s)\}$.
 \State Find all points $c_\sigma$ with the lowest $z$-coordinate in $\mathcal{U}_\sigma$.
 \State Store them as the $\sigma$-constrained geodesic centers with its $z$-value.
 \EndFor
 \State \Return All $c_\sigma$'s having the smallest $z$-value as $\cen(\Poly)$, and its $z$-value as $\rad(\Poly)$.
\EndProcedure
\end{algorithmic}
\end{minipage}
}
\caption{An $O(n^{12+\epsilon})$-time algorithm for computing $\cen(\Poly)$ and $\rad(\Poly)$
 of a polygonal domain $\Poly$.}
\label{fig:alg_spm}
\end{figure}

\begin{lemma} \label{lemma:dist_func_spm}
For any cell $\sigma$ of $\DecompSPM$ and $i\in \{1,2,\ldots, m\}$, it holds that
 \[ f_i(s) = 
   \begin{cases}
      \|s - v_i(s) \|  & \text{if $v_i$ is visible from $\sigma$}, \\
      \plf_{u_i, w_i}(s, v_i(s)) & \text{otherwise},
    \end{cases}\]
 where $u_i, w_i \in V$ are two corners of $\Poly$ uniquely determined by $i$.
\end{lemma}
\begin{proof}
The case in which $v_i$ is visible follows from Lemma~\ref{lemma:visible_vertex_spm}.
Thus, it suffices to consider the opposite case. Pick any point $s_0 \in \sigma$,
and consider a shortest path $\pi$ from $s_0$ to $v_i(s_0)$.
Let $u_i\in V$ and $w_i\in V$ be the first and the last corners of $\Poly$ along $\pi$.
Since $v_i$ is not visible from $\sigma$ (and in particular from $s_0$),
no shortest path from $s_0$ to $v_i(s_0)$ can be the straight line segment $\seg{s_0v_i(s_0)}$.
Therefore, such corners $u_i$ and $w_i$ must exist.
This implies that $f_i(s_0) = \dist(s_0, v_i(s_0))= \plf_{u_i, w_i}(s_0, v_i(s_0))$.
By the definition of the SPM-equivalence decomposition $\DecompSPM$,
for any $s\in \sigma$, the shortest paths from any $s\in \sigma$ to $v_i(s)$
have the same combinatorial structure.
Therefore, the path whose first corner is $u_i$ and last corner is $w_i$ must also be a shortest path.
Hence, the lemma follows.
\end{proof}

By combining these two observations, we can explicitly construct the functions $f_1, f_2, \ldots, f_m$,
and exploit them to compute the geodesic radius and center.
The pseudocode of our algorithm can be found in \figurename~\ref{fig:alg_spm}.

\begin{theorem} \label{theorem:algo_spm}
 The algorithm described in \figurename~\ref{fig:alg_spm} correctly computes
 the geodesic radius and center of a polygonal domain with $n$ corners
 in $O(n^{12+\epsilon})$ time for any $\epsilon>0$.
\end{theorem}
\begin{proof}
The correctness follows from the discussion above.
That is, any center of $\Poly$ corresponds to a minimum of the upper envelope of the functions $f_1, f_2, \ldots, f_m$.

In order to show the time bound, we need an efficient tool to compute the upper envelope of functions.
Given a collection of $N$ algebraic surface patches in $\Real^d$,
we can compute their lower (or upper) envelope in $O(N^{d-1+\epsilon})$ time
using the algorithms of Halperin and Sharir~\cite{hm-nbletdavt-94} (for $d=3$)
or of Sharir~\cite{s-atublehd-94} (for $d>3$).
Note that the complexity of the resulting envelope is bounded by $O(N^{d-1+\epsilon})$.

Recall that the coordinates of each vertex $v_i(s)$ of $\SPM(s)$ is
an algebraic function~\cite{cm-tpespqp-99}.
Lemma~\ref{lemma:dist_func_spm} implies the functions $f_i$ are algebraic, too.
Thus, we can apply the above algorithms to compute
the upper envelope $\mathcal{U}_\sigma$ of the graphs of $f_i$, and obtain an explicit expression of $\prad$.

In our case, we have $N = m = O(n)$, since any shortest path map $\SPM(s)$ has $O(n)$ complexity.
Each function $f_i$ has two degrees of freedom (i.e., the coordinates of $s$ within $\sigma$),
so the graph of $f_i$ lies in three-dimensional space.
That is, the upper envelope $\mathcal{U}_\sigma$ of the functions $f_i$ can be computed in $O(n^{2+\epsilon})$
for any positive $\epsilon$.
Once the envelope is computed, we can find the points with the lowest $z$-coordinate in $\mathcal{U}_\sigma$
in the same time bound by traversing all faces of the envelope $\mathcal{U}_\sigma$.
Any point that minimizes $\mathcal{U}_\sigma$ is a candidate for a geodesic center,
and its image will be its corresponding radius.

Consequently, we spend $O(n^{2+\epsilon})$ time per cell $\sigma$ of $\DecompSPM$.
Since $\DecompSPM$ consists of $O(n^{10})$ cells, we obtain the claimed time bound $O(n^{12+\epsilon})$.
\end{proof}

\section{Concluding Remarks} \label{sec:conclusion}
We have presented the first algorithm that computes
the geodesic radius and center of a general polygonal domain with holes.
The running time of our algorithm is large, but still comparable with those for other related problems.
A bottleneck of our algorithm is to compute the SPM-equivalence decomposition $\DecompSPM$,
which is very complicated and not well understood.
The best known upper bound on the complexity of $\DecompSPM$ is $O(n^{10})$,
and it is known how to construct a polygonal domain whose decomposition $\DecompSPM$
has $\Omega(n^4)$ complexity~\cite{cm-tpespqp-99}.
Thus, a better analysis on the upper bound for $\DecompSPM$
would directly lead to an improvement to our algorithm.

Another approach for improvement would be to use a coarser subdivision,
such as the \emph{SPT-equivalence decomposition}~\cite{cm-tpespqp-99}.
The SPT-equivalence decomposition $\DecompSPT$ only requires shortest path trees $\SPT(s)$
for all $s$ in each cell of $\DecompSPT$ to be isomorphic, rather than equivalence between $\SPM(s)$.
The complexity of this subdivision is $O(n^4)$, which is much smaller than that of $\DecompSPM$,
and has similar (albeit slightly weaker) properties to those of $\DecompSPM$.
Ideally, we would want an algorithm that can compute the $\sigma$-constrained geodesic radius
for a cell $\sigma$ of $\DecompSPT$ in $o(n^{8+\varepsilon})$ time
so that overall the running time improves Theorem~\ref{theorem:algo_spm}.
However, all of our attempts needed significantly more than $\Omega(n^{8})$ time,
which lead to even slower algorithms.

Throughout the paper, we have focused on the exact computation of the geodesic radius and centers,
but one could also consider the approximation variant.
By the triangular inequality, any point $s\in \Poly$ and its farthest point $t\in\Poly$
give a $2$-approximation of the radius.
That is, $\rad(\Poly) \leq \prad(s) \leq 2\, \rad(\Poly)$ for any $s\in\Poly$.
Similarly, we can obtain a $(1+\varepsilon)$-approximation by using a standard grid technique:
Scale $\Poly$ so that $\Poly$ fits into a unit square,
and partition $\Poly$ with a grid of size $O(\varepsilon^{-1})\times O(\varepsilon^{-1})$.
Define the set $D$ to be the point set containing the grid points that are inside $\Poly$,
and intersection points between boundary edges and grid edges.
We then observe that the distance between any two points $s$ and $t$ in $\Poly$ is
within a $(1+\varepsilon)$-factor of the distance between two points of $D$ that are closest
from $s$ and $t$, respectively.
Hence, we conclude that $\min_{z \in D} \prad(z) \leq (1+\varepsilon)\cdot \rad(\Poly)$.
Since $D$ consists of $O(\varepsilon^{-1} ( \varepsilon^{-1} + n))$ points
and it takes $O(n\log n)$ time per each point $z\in D$ to find its farthest neighbors
using the shortest path map $\SPM(z)$,
this algorithm runs in $O((\frac{n}{\varepsilon^2}+\frac{n^2}{\varepsilon})\log n)$ time.\footnote{%
A similar approach for approximating the geodesic diameter was mentioned in~\cite{bko-gdpd-13}.}
No subquadratic-time approximation algorithm with factor less than $2$
is known so far.
%

{ 
\bibliographystyle{abbrv}
\bibliography{pd}
}

\end{document}